\documentclass[twocolumn]{IEEEtran} 

\usepackage{amsmath,amssymb,amsfonts}
\usepackage[final]{graphicx}
\usepackage{psfrag}
\usepackage[tight]{subfigure}
\usepackage[numbers,sort&compress]{natbib}
\usepackage{etoolbox}
\makeatletter
\patchcmd{\@makecaption}
  {\scshape}
  {}
  {}
  {}
\makeatletter
\patchcmd{\@makecaption}
  {\\}
  {.\ }
  {}
  {}
\makeatother

\usepackage{mathrsfs} 
\usepackage{epstopdf}
\usepackage{multirow}
\usepackage{amsthm,mathtools}
\usepackage{float}
\usepackage{booktabs}  
\usepackage{multirow} 
\usepackage{relsize}
\usepackage{bm}
\usepackage{hyperref}
\usepackage[nolist,nohyperlinks]{acronym}
\usepackage{svg}
\usepackage{xparse}
\NewDocumentCommand{\evalat}{sO{\big}mm}{%
  \IfBooleanTF{#1}
   {\mleft. #3 \mright|_{#4}}
   {#3#2|_{#4}}%
}
\newcommand{\norm}[1]{\left\lVert#1\right\rVert}

\begin{acronym}
\acro{PLS}{physical layer security}
\acro{OP}{outage probability}
\acro{SNR}{signal-to-noise ratio}
\acrodefplural{SNR}[SNRs]{signal-to-noise ratios}
\acro{AWGN}{additive white Gaussian noise}
\acro{CSI}{channel state information}
\acro{PDF}{probability density function}
\acrodefplural{PDF}[PDFs]{probability density functions}
\acro{CDF}{cumulative distribution function}
\acrodefplural{CDF}[CDFs]{cumulative distribution functions}
\acro{LRS}{large reflecting surface}
\acro{BS}{base station}
\acro{RF}{radio frequency}
\acro{RV}{random variable}
\acro{FN}{folded normal}
\acro{TDD}{time-division duplexing}
\acro{LOS}{line-of-sight}
\acro{DL}{downlink}
\acro{MRT}{maximal ratio transmission}
\acro{MRC}{maximal ratio combining}
\acro{MC}{Monte Carlo}
\acro{RV}{random variable}
\acro{CDF}{cumulative density function}
\acro{FAS}{fluid antenna system}
\acro{MIMO}{multipe-input multiple-output}
\acro{SNR}{signal-noise-to-ratio}
\acro{BS}{base station}
\acro{CSI}{channel state information}
\end{acronym}

\newtheorem{proposition}{Proposition}

\newtheorem{remark}{Remark}

\usepackage[numbers,sort&compress]{natbib}

\makeatletter
\def\blfootnote{\xdef\@thefnmark{}\@footnotetext}
\makeatother

\DeclarePairedDelimiter\abs{\lvert}{\rvert} 
\usepackage{float}
\usepackage{stfloats}
\usepackage{textcomp}
\usepackage[ruled,vlined]{algorithm2e}

\begin{document}
\title{\huge{A Simple Method for the Performance Analysis of Fluid Antenna Systems under Correlated Nakagami-$m$ Fading }}
\author{Jos\'e~David~Vega-S\'anchez, \textit{Member, IEEE}, Luis~Urquiza-Aguiar, \textit{Member, IEEE}, Martha Cecilia Paredes Paredes, \textit{Senior Member, IEEE}, and Diana~Pamela~Moya~Osorio, \textit{Senior Member, IEEE} }

\maketitle

\blfootnote{\noindent Manuscript received MONTH xx, YEAR; revised XXX. The review of this paper was coordinated by XXXX. The work of Luis Urquiza-Aguiar was supported by the Escuela Polit\'ecnica Nacional. The work of D. P. M. Osorio is supported by the Academy of Finland, project FAITH under Grant 334280.}

\blfootnote{\noindent
J.~D.~Vega-S\'anchez is with the Faculty of Engineering and Applied Sciences (FICA), Telecommunications Engineering, Universidad de Las Am\'ericas
(UDLA), Quito 170124, Ecuador (E-mail: $\rm jose.vega.sanchez@udla.edu.ec$), and also with the Departamento de Electr\'onica, Telecomunicaciones y Redes de Informaci\'on, Escuela Polit\'ecnica Nacional (EPN),
Quito,  170525, Ecuador. 
}
\blfootnote{\noindent 
 L.~Urquiza-Aguiar and M. C. Paredes Paredes are with the  
Departamento de Electr\'onica, Telecomunicaciones y Redes de Informaci\'on, Escuela Polit\'ecnica Nacional (EPN), Ladr\'on de Guevara E11-253, Quito,  170525, Ecuador. (e-mail: cecilia.paredes@epn.edu.ec; luis.urquiza@epn.edu.ec).
}

\blfootnote{\noindent
D.~P.~Moya~Osorio is with the Centre for Wireless Communications (CWC), University of Oulu, Finland (e-mail: diana.moyaosorio@oulu.fi)
}



\vspace{-12.5mm}
\begin{abstract}
By recognizing the tremendous flexibility of the emerging fluid antenna system (FAS), which allows dynamic reconfigurability of the location of the antenna within a given space, this paper investigates the performance of a single-antenna FAS over spatially correlated Nakagami-$m$ fading channels. Specifically,  simple and highly accurate closed-form approximations for the cumulative density function of the FAS channel and the outage probability of the proposed system are obtained by employing a novel asymptotic matching method, which is an improved version of the well-known moment matching. With this method, the outage probability can be computed {simply} without incurring complex multi-fold integrals, thus requiring negligible computational effort. Finally, the accuracy of the proposed approximations is validated, and it is shown that the FAS can meet or even exceed the performance attained by the conventional maximal ratio combining (MRC) technique.
\end{abstract}


\vspace{-2.5mm}
\section{Introduction}
The fifth-generation (5G) of wireless mobile networks have recently been deployed worldwide, so industry and academia have already started the race to define the shape red {of the} future sixth-generation (6G). Very recently, a technology that has been gaining momentum is the \ac{FAS}, which is a new paradigm of antenna systems where antennas are equipped with software-controllable fluid structure  (e.g., Eutectic Gallium-Indium, Mercury, Galinstan, etc.) that allows {dynamic} reconfigurability of its position and shape within a given space{\footnote{{Interested readers can refer to \cite{Ana} for information on fluid antenna prototypes.}}}. Particularly, FAS may help overcome practical limitations of {using multiple antennas} in size-constrained devices, and the cost of \ac{RF} chains \cite{Wong,bruce}. The fundamental single fluid antenna is built of one \ac{RF} chain and $N$ fixed locations, so-called ports, distributed in a linear space. Unlike conventional {spatial diversity techniques (e.g., \ac{MRC})}, FAS allows an antenna to freely switch its position among the ports to obtain a more robust channel gain or lower interference, thus providing {remarkable} gains in diversity{,} multiplexing, 
 and interference-free communications \cite{New}.


A plethora of works have been recently
focused {on investigating} the performance of FAS in wireless communications, where metrics such as ergodic capacity and \ac{OP} metrics have been investigated {in} different settings.
For instance, in \cite{Yangyang}, Wong et al. demonstrated that a single-antenna FAS outperforms the traditional \ac{MRC} in terms of the \ac{OP} when the number of ports at the fluid antenna is large enough. Also{,} in~\cite{Wong,correlation}, Wong et al. studied the achievable performance of FAS in arbitrarily correlated Rayleigh fading channels.  Khammassi et al. proposed an approximate expression for the FAS relative channel distribution in \cite{Ref3}, where a two-stage approach was proposed. The first phase reduces the number of multi-fold integrals of the OP, while the second represents the OP in a single-integral form by assuming correlated Rayleigh fading channels. Tlebaldiyeva et al. considered a more general small-scale fading channel model on the FAS, in \cite{Tlebaldiyeva1}, where the OP was found in a single-integral form for a single-antenna $N$-port FAS over spatially Nakagami-$m$ channel. In \cite{Ref4}, by taking advantage of stochastic geometry tools,  Skouroumounis and  Krikidis derived a closed-form expression for the OP in fluid antenna for large-scale cellular networks. Moreover, Ghadi et al. derived a closed-form formulation of the OP performance in \cite{Ref5} by adopting copula theory to characterize the correlation model (e.g., Frank, Clayton, and Gumbel) between fading channel coefficients. Very recently,  the OP behavior of FAS-aided Terahertz communication networks under correlated $\alpha$-$\mu$ fading channels for non-diversity and diversity FAS receivers was investigated by Tlebaldiyeva et al. in \cite{Tlebaldiyeva2}. Therein, as in \cite{Tlebaldiyeva1}, the OP was derived in single-integral expression due to the mathematical intractability of both the $\alpha$-$\mu$ channel model and the diversity FAS underlying system.

Based on the above considerations and motivated by the potential of FAS
to provide diversity and remarkable capacity benefits for forthcoming networks, this work {exploits} the advantages of a novel asymptotic matching method to approximate the  FAS’s channel distribution into a single approach. {Despite the FAS system's intricacy, the authors aim to provide analytically tractable expressions for the outage metric without incurring the prohibitive complexity of special functions or multi/single-fold integrals, which have already been used in previous works.} Specifically, a FAS that experiences correlated Nakagami-$m$ fading channels is considered, where we propose to approximate the equivalent \ac{CDF} of the FAS by a simple Gamma distribution, where the fitting parameters are estimated via the asymptotic matching method, proposed in \cite{Perim}. With the \ac{CDF} of FAS at hand, we derive a simple and highly accurate closed-form expression of the \ac{OP} valid to the practical correlation models introduced in \cite{correlation}. {It is worth mentioning that such a method outperforms those approaches based solely on the moment matching method (MoM)\footnote{As stated in \cite{wang2003simple}, some performance metrics in communications systems, such as the OP or bit error rate, are dominated by the channel asymptote at medium to high \ac{SNR}. So, the asymptotic matching method, compared to the MoM, provides an excellent fit in medium to high \ac{SNR}, which is a crucial regime in practice. Conversely, the MoM delivers good performance at low \ac{SNR}, a range of little importance in practical applications.} and improves the computational complexity of cumbersome traditional exact formulations. Compared to the latter, obtained in terms of the correlated join \ac{PDF}\footnote{{This solution becomes tedious as the number of ports increases.}}, our approach presents a simple form of the \ac{OP} in terms of well-known functions in the communication theory community, facilitating its numerical evaluation in any computer software.} Finally, useful insights into the impact of propagation conditions and the number of ports over the \ac{OP} performance are also provided.
\vspace{-4.5mm}
\section{System and Channel Models}
\begin{figure}[t]
\centering
\psfrag{J}[Bc][Bc][0.7]{$\mathrm{Transmitter}$}
\psfrag{K}[Bc][Bc][0.7]{$\mathrm{FAS-receiver}$}
\psfrag{F}[Bc][Bc][0.7]{$1$}
\psfrag{G}[Bc][Bc][0.7]{$2$}
\psfrag{H}[Bc][Bc][0.7]{$3$}
\psfrag{I}[Bc][Bc][0.7]{$N$}
\psfrag{L}[Bc][Bc][0.7]{$\mathrm{Ports}$}
\psfrag{P}[Bc][Bc][0.7]{$\mathrm{Fluid}$}
\psfrag{M}[Bc][Bc][0.7]{$\mathrm{ antenna}$}
\psfrag{O}[Bc][Bc][0.7]{$\mathrm{switch}$}
\psfrag{B}[Bc][Bc][0.7]{$\mathrm{radio~chip}$}
\psfrag{A}[Bc][Bc][0.7]{$\mathrm{fluid~container}$}
\psfrag{C}[Bc][Bc][0.7]{$\mathrm{RF~chain}$}
\psfrag{E}[Bc][Bc][0.7]{$\max (\abs{g_k}{})$}
\psfrag{D}[Bc][Bc][0.7][90]{$W\lambda$}
 \includegraphics[width=0.8\linewidth]{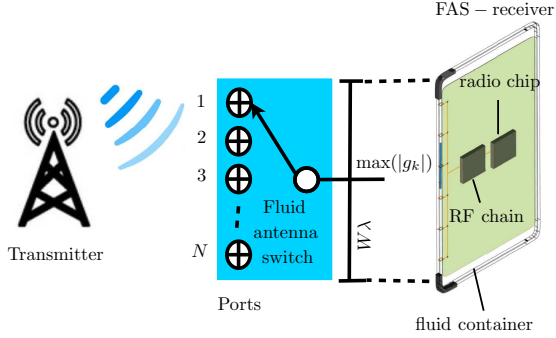}
\caption{System Model for single-antenna $N$-port FAS.}
\label{SM}
\vspace{-4.5mm}
\end{figure}
We consider a FAS consisting of a single-antenna transmitter (Tx)
communicating with a fluid single-antenna that can move freely along 
$N$ ports equally distributed on a linear space of length
$W\lambda$, where $\lambda$ is the wavelength{,} and $W$ is the antenna size, as illustrated in Fig. \ref{SM}. The FAS contains one \ac{RF} chain, thus only one port  can be activated for communication, so the received signal at the $k$-th port can be defined as 
\begin{equation}
\label{eq1}
Y_k=h_kX+Z_k,
\end{equation}
with $h_k$ being modeled as a correlated Nakagami-$m$ fading channel since antenna ports are located close to each other. Also, $X$ is the information signal, and $Z_k$ is the \ac{AWGN} at every port. We assume that FAS can switch the port with the strongest signal for communication, which can be expressed as
\begin{equation}
\label{eq2}
g_{\mathrm{FAS}} =\max \left ( \left | g_1   \right |, \left | g_2\right |, \cdots, \left | g_k\right | \right ), \hspace{2mm} \textit{for}\hspace{2mm} k \in  
  \left\{1,2,\cdots, N \right\},
\end{equation}
where $g_{i}= \left | h_{i}\right |^{2}$ for $i\in\left \{ 1,\cdots,N \right \}$ denotes the channel gain of each port in the \ac{FAS}.
Hence, the received \ac{SNR}, for the FAS can be expressed by
\begin{equation}\label{eq3}
\gamma=\frac{P \left | g_{\mathrm{FAS}}\right | }{N_0}= \overline{\gamma} \left | h_{\mathrm{FAS}}\right |, 
\end{equation}
where $\overline{\gamma}=\frac{P}{N_0}$ is the average transmit \ac{SNR}{,} with $P$ being the transmit power and $N_0$ the noise power. Considering this, {we aim} to provide an approximate statistical model for $g_{\mathrm{FAS}}$, which can be used to obtain the \ac{OP} distribution straightforwardly.
\section{Performance Analysis}
In this section, we consider the \ac{OP} to
evaluate the performance of the FAS,  which is defined as the probability
that the \ac{SNR} $\gamma$ is less than a threshold rate, $\gamma_{th}$. Therefore, from~\cite[Eq.~(10)]{Tlebaldiyeva1}, the \ac{OP} formulation over correlated Nakagami-$m$ random variables (RVs) can be formulated as
\begin{align}\label{eq4}
P_{\mathrm{out}}(\gamma_{th})&= \frac{2^m m }{\Gamma(m)\Omega_1^{2m} }\int_{0}^{\sqrt{\frac{\gamma_{th}}{\overline{\gamma}}}}r_1^{2m}\exp\left ( -\frac{m r_1^2}{\Omega_1^{2}} \right )\nonumber \\ \times & \prod_{k=2}^{N}\left ( 1-Q_m\left ( \sqrt{\tfrac{2m\mu_k^2r_1^2}{\Omega_1^{2}\left ( 1-\mu_k^2 \right )}} \right ),\sqrt{\tfrac{2m\gamma_{th}}{\Omega_k^{2}\left ( 1-\mu_k^2 \right )\overline{\gamma}}} \right )dr_1,
\end{align}
where $\Gamma(\cdot)$ is the Gamma function, $Q_m(\cdot,\cdot)$ denotes the $m$-order Marcum Q-function, $m$ is the
severity fading parameter, and $\Omega_k^2 $ stands for the average channel power of Nakagami-$m $ distribution. Furthermore,  the correlation coefficient, denoted by, $\mu_k$, can be defined as~\cite[Eq.~(2)]{Tlebaldiyeva1} when the $(N-1)$ ports are referenced to the first port or can be expressed as~\cite{correlation} \begin{equation}\label{eq5}
\mu^2 =\left | \frac{2}{N(N-1)}\sum_{k=1}^{N-1}(N-k)J_0\left ( \frac{2\pi kW}{N-1} \right )\right |, \hspace{1mm} \textit{for}\hspace{1mm} \mu_k=\mu~ \forall k, 
\end{equation}
where it is assumed that all the ports do not have a reference port or any port is a reference to any other port, and $J_0(\cdot)$
denotes the zero-order Bessel function of the first kind. It is worth highlighting that as the number of ports increases, the exact solution in~\eqref{eq4} becomes more costly, prone to convergence and instability problems, or even impracticable. To circumvent the referred limitation of the exact \ac{OP}, an approximation for~\eqref{eq4} is proposed by using the asymptotic matching method \cite{Perim}{,} as stated in the following proposition.
\begin{proposition}\label{Propos1}
An approximate expression for the \ac{OP} of a FAS undergoing correlated Nakagami-$m$ RVs can be obtained as
\begin{align}\label{eqs6}
P_{\mathrm{out}}(\gamma_{th})&\approx \frac{\Upsilon(\alpha,\frac{\gamma_{th}}{\beta \overline{\gamma}})}{\Gamma{\left ( \alpha \right )}},
\end{align}
where is $\Upsilon(\cdot,\cdot)$, is the lower incomplete gamma function~\cite[Eq.~(6.5.2)]{Abramowitz} and
\begin{subequations}
\label{eqs7}
	\begin{align}
	\label{eq7a}
	\alpha  =& mN, \hspace{3mm} \beta=\left ( \frac{1}{\Gamma(\alpha)a_0 \alpha } \right )^{1/\alpha},
	\\
	\label{eq7b}
	a_0 =&\frac{m^{m-1}}{\Gamma(m)\Omega_1^{2m} m!^{N-1}}\prod_{k=2}^{N}\left ( \frac{m}{\Omega_k^2\left ( 1-\mu_k^2 \right )} \right )^m.
	\end{align}
\end{subequations}
\end{proposition}
\begin{proof}
See Appendix~\ref{appendix1}.
\end{proof}
\begin{remark}\label{remark1}
Unlike previous contributions in the literature (e.g.,~\cite{Tlebaldiyeva1,correlation,Tlebaldiyeva2}), where the \ac{OP} was derived in integral form, the result in~\eqref{eqs7} is a simple and accurate approximation that does not need to solve any involved integrals regarding the joint distribution of correlated fading channels of FAS in order to reach the OP metric. Moreover,~\eqref{eqs7} is valid for the practical spatial correlation models proposed in~\cite{correlation}.
\end{remark}
{An asymptotic closed-form expression for the \ac{OP} is derived to gain more insight into the impact of system parameters on the FAS performance.} For that purpose, the asymptotic \ac{OP} is expressed {as} $\mathrm{OP}^{\infty}\simeq \mathrm{G}_c\overline{\gamma}^{-\mathrm{G}_d}$~\cite{wang2003simple}, where $\mathrm{G}_c$ and $\mathrm{G}_d$ {are} the array gain and the diversity order, respectively. The asymptotic \ac{OP} is given in the following Proposition.
\begin{proposition}\label{Propos2}
The asymptotic \ac{OP} expression for the proposed FAS undergoing correlated Nakagami-$m$ RVs is given by
\begin{align}\label{eq8}
P_{\mathrm{out}}(\gamma_{th})&\simeq  \frac{(\frac{\gamma_{th}}{\beta \overline{\gamma}})^{\alpha}}{\alpha \Gamma{\left ( \alpha \right )}},
\end{align}
\end{proposition}
\begin{proof}
To asymptotically approximate~\eqref{eqs6}, the  relationship $\Upsilon \left (a,x \right )\simeq x^a/a  $ as $x\rightarrow 0$ in~\eqref{eqs6}, is employed.
\end{proof}
\begin{remark}\label{remark2}
From~\eqref{eq8}, it is evident that the diversity order, $\mathrm{G}_d=Nm$, is directly affected by the number of ports and the severity of fading.
\end{remark}

\section{Numerical results and discussions} \label{sect:numericals}
\begin{figure}[t]
\psfrag{A}[Bc][Bc][0.7]{$\mathrm{\textit{n}=49}$}
\psfrag{B}[Bc][Bc][0.7]{$\mathrm{\textit{n}=100}$}
\psfrag{C}[Bc][Bc][0.7]{$\mathrm{\textit{n}=196}$}
 \includegraphics[width=\linewidth]{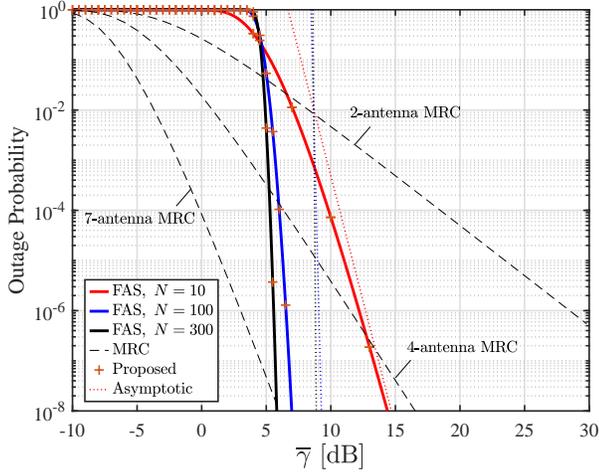}
\caption{\ac{OP} vs. $\overline{\gamma}$, for different numbers of ports of FAS by assuming $W=0.3$, $m=1$, and $\gamma_{th}=1$ dB. Markers denote the proposed approximation given in~\eqref{eqs6}, whereas the solid and dotted lines represent the analytical and the asymptotic solutions given in~\eqref{eq4} and~\eqref{eq8}, respectively.}
\label{fig1}
\vspace{-4.5mm}
\end{figure}
In this section, the effect of the system model parameters (e.g., {antenna size}, number of ports, and the severity of fading) on the \ac{OP} performance in  a FAS is addressed, as well as the goodness of the proposed approximation for the equivalent channel. Unless stated otherwise, for all figures, it is assumed that $\Omega_k=1, \forall k$, and the spatial correlation model is formulated from~\eqref{eq5}. For the sake of comparison, the conventional \ac{MRC} technique with uncorrelated antennas is
included as a reference in the \ac{OP} analysis.

\begin{figure}[h!]
\psfrag{A}[Bc][Bc][0.7]{$\mathrm{\textit{n}=49}$}
\psfrag{B}[Bc][Bc][0.7]{$\mathrm{\textit{n}=100}$}
\psfrag{C}[Bc][Bc][0.7]{$\mathrm{\textit{n}=196}$}
 \includegraphics[width=\linewidth]{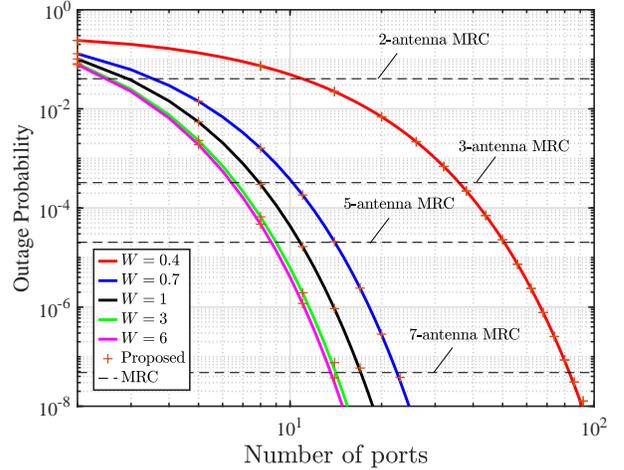}
\caption{\ac{OP} vs. number of ports by varying $W$ for $\overline{\gamma}=5$ dB, $m=1$, and $\gamma_{th}=1$ dB. Solid lines denote the exact solution in~\eqref{eq4}.}
\label{fig2}
\vspace{-4.5mm}
\end{figure}

In Fig~\ref{fig1}, we illustrate the \ac{OP} against the $\overline{\gamma}$ for different numbers of ports of FAS by setting $W=0.3$, $m=1$, and $\gamma_{th}=1$ dB. In this case, the purpose is to demonstrate the accuracy of~\eqref{eqs6} to approximate the exact \ac{OP} solution given in~\eqref{eq4}. Note that all figures show that the proposed approximations agree remarkably with the analytical ones for the whole average \ac{SNR} range. It is worthwhile to mention that~\eqref{eq4} is very demanding to calculate as the number of ports increases. Therefore, our approach with a simple mathematical fashion and negligible computational effort is highly attractive for further developments of FAS. On the other hand, based on the asymptotic plots, it is observed that the massive diversity of FAS contributes to the \ac{OP} slope {proportionally}. This means that the OP decay is steeper (i.e., better \ac{OP} performance is obtained) when the
number of ports or the $m$ parameter is increased (i.e., soft fading). Conversely, the \ac{OP} is impaired as the number of ports or $m$ decreases, and the decay is not so pronounced. These facts are in coherence with the results discussed in Remark~\ref{remark2}. Note that the asymptotic \ac{OP} quickly reaches the diversity order for $N=10$. Contrariwise, for the cases with $N=100,300$, the asymptotic \ac{OP} matches the true asymptotic behavior for extremely low operational OP values. {In addition, a crossover in the OP is exhibited when $\overline{\gamma}<5$ dB; to better understand this behavior, an important remark is in order. Let us assume the case where the OP operates without fading, i.e., only the \ac{AWGN} channel is considered. In this context, based on \cite{Simon2005}, the OP is equal to $1$ for the \ac{SNR} values below the outage threshold and identical to $0$ otherwise. Now, by assuming fading channels (e.g., Nakagami-$m$), the ${\rm OP}<1$ for those \ac{SNR} values below the outage threshold, whereas ${\rm OP}>0$ in the range of \ac{SNR} values above such threshold. In other words, the slope of the OP does not decay as abruptly because of fading fluctuations (e.g., see the MRC curves). Based on this, for the Nakagami-$m$ channel, as fading severity is soft (i.e., large $m$ values), OP curves tend to reveal an OFF/ON behavior, similar to the \ac{AWGN} case. Now, from Remark~\ref{remark2}, it is clear that the diversity order of the OP is proportional to fading severity and the number of ports, $N$. This explains the behavior of the curves below $\overline{\gamma}<5$ dB, i.e., for large $N$ values, the OP plots behave like an AWGN channel, creating such a crossover. However, it is worth mentioning that, this behavior occurs in a non-operational range of SNR and OP of a standard communication system.} Finally, the \ac{OP} performance of FAS can meet or even exceed the \ac{MRC} for high $N$ values. {In fact, this performance gap could be even more prominent in real-world rich-scattering environments with a large number of multipath clusters (e.g., an indoor area or a city street canyon). This is because the slope of the OP is governed by the $Nm$ term, which means that large values of $N$ or $m$ (i.e., rich-scattering) lead to remarkable performance gains compared to classic MRC.}

\begin{table}[t]
	\centering
	\textsf{
	\caption{ {Comparison of computational efforts between the proposed approach and exact solution. \\NMSE ranges from $-\infty$ (bad fit) to 1 (perfect fit).}} \label{elapsedtime}
	\centering
   \tiny
	\begin{tabular}{ccccccc}
		\toprule
		\multicolumn{1}{c}{\multirow{3}{*}{}} &  &\multicolumn{2}{c}{\textbf{{Average Elapsed Time (seconds)}}} & & &\textbf{{NMSE}}\\
		\cmidrule(lr){3-5} \cmidrule(lr){7-7} & 
		 &      \textbf{{Proposed}}  &  \textbf{{Proposed}} &  \textbf{{Exact}} & \textbf{{Time }} & \textbf{{Proposed}} \\ \textbf{{Fig.~\#}}& \textbf{{Case}}
		  &  \textbf{{Asymptotic OP}}  &  \textbf{{ OP}} &  \textbf{{OP}} & \textbf{{Reduction (\%)}}  & \textbf{{OP}}\\
		\cmidrule(lr){1-7} 
		& \multicolumn{1}{c}{\textbf{{$N=10$}}}& {0.0312}  & {0.0781 }   &{ 24.554}  &{99.68} & {0.99}\\
		\cmidrule(lr){2-7}
		&\multirow{1}{*}{\textbf{{$N=100$}}}  &{ 0.1406} & {0.6562} & {141.90} &{99.53} & {0.99}\\
		\cmidrule(lr){2-7}
		&\multirow{1}{*}{\textbf{{$N=300$}}}  &{0.4375} &{2.0937} &{391.26} & {99.46} &{0.99}\\
	   \cmidrule(lr){2-7}
	      	\multicolumn{1}{c}2
        {\multirow{3}{*}{}} &  &\multicolumn{2}{c}{\textbf{{Memory in Use (megabytes)}}} & & &\\
		\cmidrule(lr){3-5} & 
		 &      \textbf{{Proposed}}  &  \textbf{{Proposed}} &  \textbf{{Exact}} & \textbf{{Memory }} & \\ \textbf{{}}& \textbf{{Case}}
		  &  \textbf{{Asymptotic OP}}  & \textbf{{OP}} &  \textbf{{OP}} & \textbf{{Reduction (\%)}}  &\\
    \cmidrule(lr){2-6} 
		& \multicolumn{1}{c}{\textbf{{$N=10$}}}& {0.0095}  & {0.0211 }   &{3.5810}  &{99.41} &\\
		\cmidrule(lr){2-6}
		&\multirow{1}{*}{\textbf{{$N=100$}}}  &{0.0121} & {0.0334} & {3.6123} & {99.07} & \\
		\cmidrule(lr){2-6}
		&\multirow{1}{*}{\textbf{{$N=300$}}}  &{0.0238} & {0.0492} &{3.6726} & {98.66} & \\
	   \cmidrule(lr){1-7}
       \end{tabular}}
       \vspace{-4mm}
\end{table}
\begin{figure}[t]
\psfrag{A}[Bc][Bc][0.7]{$\mathrm{\textit{n}=49}$}
\psfrag{B}[Bc][Bc][0.7]{$\mathrm{\textit{n}=100}$}
\psfrag{C}[Bc][Bc][0.7]{$\mathrm{\textit{n}=196}$}
 \includegraphics[width=\linewidth]{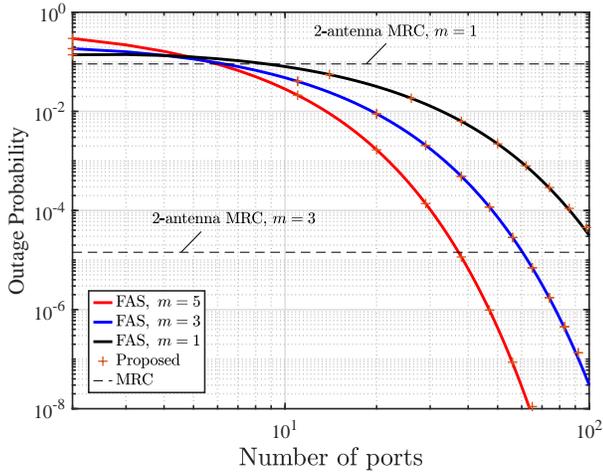}
\caption{\ac{OP} vs. number of ports with $m=\left \{ 1,3,5 \right \}$ for $W=0.6$, $\overline{\gamma}=3$ dB, and $\gamma_{th}=1$ dB. Solid lines denote the exact solution in~\eqref{eq4}.}
\label{fig3}
\vspace{-4.5mm}
\end{figure}
In Fig.~\ref{fig2}, the \ac{OP} is illustrated as a function of the number of ports of FAS for $\overline{\gamma}=5$ dB, $m=1$, and $\gamma_{th}=1$ dB. In these curves, we explore the impact of varying the antenna size over the \ac{OP} performance. It can be observed that large-size $W$ coefficients (i.e., more space in the FAS) {result in} a better \ac{OP} performance, as expected. Hence, the performance of FAS hugely relies on both the size $W$  and the number of ports $N$ of the fluid antenna at the FAS receiver.

Also, note that the \ac{OP} performance is not restricted by $W$, i.e., no outage floor exists as $N$ increases. On the other hand, FAS outperforms the MRC diversity scheme when the system is deployed with large antennas and a massive amount of ports.

In Fig.~\ref{fig3}, the \ac{OP} is illustrated as a function of the number of ports of FAS for $W=0.6$, $\overline{\gamma}=3$ dB, and $\gamma_{th}=1$ dB. Herein, {the achievable \ac{OP} is investigated} by varying the severity of fading, i.e., $m=\left \{ 1,3,5 \right \}$. Three different scenarios for the \ac{OP} behavior are observed regarding the configuration of parameters. Specifically, decreasing $m$ (i.e., hard fading) harms the performance of the \ac{OP}. Conversely, increasing $m$ (i.e., mild fading condition) favors the \ac{OP}. Furthermore, 2-antenna \ac{MRC} for $m=3$ has the lowest OP until FAS reaches $N=37$ with $m=5$. This \ac{MRC} supremacy may be due to the power gain attributed to the active \ac{RF} chains, while FAS has only one active \ac{RF} chain.

Finally, Fig.~\ref{fig4} depicts the OP versus $\gamma_{th}$ by varying $N$ for $W=2$, $\overline{\gamma}=0$ dB, and $m=1$. It can be observed that increasing $\gamma_{th}$ leads to a significant loss in the OP. The best OP performance is achieved with massive ports in the FAS. As in the previous figure, the MRC scheme (i.e., 12-antenna configuration) outperforms the FAS behavior when dealing with higher $\gamma_{th}$ values. However, this result can be reversed when the FAS combines a large antenna size together with many ports. 
\begin{figure}[t]
\psfrag{A}[Bc][Bc][0.7]{$\mathrm{\textit{n}=49}$}
\psfrag{B}[Bc][Bc][0.7]{$\mathrm{\textit{n}=100}$}
\psfrag{C}[Bc][Bc][0.7]{$\mathrm{\textit{n}=196}$}
 \includegraphics[width=\linewidth]{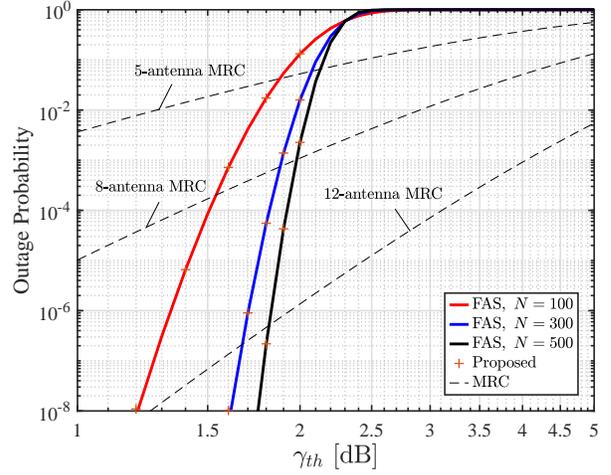}
\caption{\ac{OP} vs. $\gamma_{th}$ by varying $N$ for $W=2$, $\overline{\gamma}=0$ dB, and $m=1$. Solid lines denote the exact solution in~\eqref{eq4}.}
\label{fig4}
\vspace{-4.5mm}
\end{figure}

{To quantify the fitting accuracy, we use the widely-accepted normalized mean-square-error (NMSE) test. Specifically, the NMSE measures the goodness-of-fit between the approximate and exact OPs, denoted by $\widehat{P}_{\mathrm{out}}(\gamma_{th})$ and $P_{\mathrm{out}}(\gamma_{th})$, respectively, i.e., $\text{NMSE}=1-||P_{\mathrm{out}}(\gamma_{th})-\widehat{P}_{\mathrm{out}}(\gamma_{th})||^2/\norm{P_{\mathrm{out}}(\gamma_{th})-\mathbb{E} \left [ P_{\mathrm{out}}(\gamma_{th})  \right ]}$, where $\mathbb{E} \left [ \cdot \right ]$ and $\norm{\cdot}$ denote the expectation and the $2$-norm operators, respectively. For informative purposes, Table~\ref{elapsedtime} shows the computation time and memory in use of the proposed OP, asymptotic OP, and exact solutions  of the
illustrative examples considered in~Figs.~\ref{fig1}-\ref{fig2}\footnote{{The computation of both the exact and the proposed OP expressions have been run in Windows 10 (64-bits) Pro Intel (R) Core (TM) i7-10510U - 2.9 GHz - 16 GB RAM.}}, when they are evaluated numerically. The reader can notice that our OP proposal is faster than the exact solutions and reduces the computational effort above $99.46$\% in all the examples regarding the exact ones. Concerning the asymptotic OP, it is noticeable that the elapsed time for the complete set of examples is negligible and less than one second, influencing very slightly the number of ports in the elapsed time. Here, it can also be noted that memory in use reports an increase of more than $3$ megabytes in memory consumption of the exact formulation compared with its counterpart in all cases examined, which makes our approach more attractive to run in any computing software. Likewise, Table~\ref{elapsedtime} shows the NMSE between the approximate and exact OPs using the \texttt{goodnessOfFit} built-in function in MATLAB of the plots considered in Figs.~~\ref{fig1}-\ref{fig2}. Results in Table~\ref{elapsedtime} support the observation that our approach provides an outstanding fitting of the OP plots concerning the exact solutions, regardless of the value of $N$. }
\vspace{-5mm}
\section{Conclusions}
In this letter, we investigated the OP performance of a point-to-point FAS by assuming correlated Nakagami-$m$ fading channels. Specifically, a novel asymptotic matching method is employed to approximate the CDF of FAS {simply } without incurring multi/single fold integrals. Then, with this result, a simple closed-form expression of the OP for the underlying system was obtained. Moreover, useful insights were provided regarding how fading channel conditions and the number of ports impact the OP performance of FAS. Finally, our results can be extended to FAS diversity schemes {that still need to be} explored in the literature.
\vspace{-4mm}
\appendices
\section{Proof of Proposition \ref{Propos1} }
\label{appendix1} 
We first replace~\cite[Eq.~(3)]{Xianchang} into~\eqref{eq4}, so
\begin{align}\label{Apeq1}
P_{\mathrm{out}}(\gamma_{th})&\approx \frac{2m^ m }{\Gamma(m)\Omega_1^{2m} }\underset{I_1}{\underbrace{\int_{0}^{\sqrt{\frac{\gamma_{th}}{\overline{\gamma}}}}r_1^{2m-1}\exp\left ( -\frac{m r_1^2}{\Omega_1^{2}} \right )}}\nonumber \\ \times & \underset{I_1}{\underbrace{\prod_{k=2}^{N}\left ( \tfrac{\left ( \tfrac{m\gamma_{th}}{\Omega_k^{2}\left ( 1-\mu_k^2 \right )\overline{\gamma}}  \right )^m \exp\left (-\tfrac{m\mu_k^2r_1^2}{\Omega_1^{2}\left ( 1-\mu_k^2 \right )}  \right )}{m!} \right )dr_1}}.
\end{align}
Here, with the aid of~\cite[Eq.~(3.381.1)]{Gradshteyn}, applying a change of variables, and after some mathematical manipulations,~$I_1$ can be evaluated in exact closed-fashion as
\begin{align}\label{Apeq2}
P_{\mathrm{out}}(\gamma_{th})&\approx\frac{(m\gamma_{th})^m }{\Gamma(m)\Omega_1^{2m}\overline{\gamma}^m m!^{N-1}}\left ( \tfrac{m\gamma_{th}\left ( 1+\sum_{k=2}^{N} \tfrac{\mu_k^2}{1-\mu_k^2} \right )}{\Omega_1^2\overline{\gamma}} \right )^{-m}
\nonumber \\ & \times \Upsilon\left ( m, \tfrac{m\gamma_{th}\left ( 1+\sum_{k=2}^{N} \tfrac{\mu_k^2}{1-\mu_k^2} \right )}{\Omega_1^2\overline{\gamma}} \right )\prod_{k=2}^{N}\left ( \tfrac{m\gamma_{th}}{\Omega_k^2(1-\mu_k^2)\overline{\gamma}} \right )^m.
\end{align}
Next, by applying both relationships $i)$ $\Upsilon \left (a,x \right )\simeq x^a/a  $ as $x\rightarrow 0$, and $ii)$ $F_{g_{\mathrm{FAS}}}(x)=P_{\mathrm{out}}(x \overline{\gamma} )$
into~\eqref{Apeq2},  the asymptotic behavior of the \ac{CDF} of FAS in the form $F_{g_{\mathrm{FAS}}}(x)\simeq a_0 x^{b_0}$, can be formulated as
\vspace{-3mm}
\begin{align}\label{Apeq3}
F_{g_{\mathrm{FAS}}}(x)&\simeq\underset{a_0}{\underbrace{\frac{m^{m-1} }{\Gamma(m)(\Omega_1)^{2m}m!^{N-1}}\prod_{k=2}^{N}\left ( \frac{m}{\Omega_k^2(1-\mu_k^2)} \right )^m}} x^{  \overset{b_0}{\overbrace{mN}} }.
\end{align}
Then,~\eqref{eq4} can be approximated with a Gamma distribution by utilizing the asymptotic matching method \cite{Perim}. Therefore, the CDF and the asymptotic CDF of a Gamma approximation are respectively given by
\begin{align}\label{Apeq4}
\widetilde{F}_{g_{\mathrm{FAS}}}(x)=\frac{\Upsilon\left ( \alpha, \frac{x}{\beta} \right )}{\Gamma{(\alpha)}}, \vspace{1mm} \widetilde{F}_{g_{\mathrm{FAS}}}(x)\simeq \underset{\widetilde{a}_0}{\underbrace{\frac{1}{\beta^\alpha\alpha \Gamma(\alpha) }}} x^{\overset{\widetilde{b}_0}{\overbrace{\alpha}}}.
\end{align}
Then, by applying the asymptotic matching~\cite{Perim}, i.e., $a_0=\widetilde{a}_0$ and $b_0=\widetilde{b}_0$, the shape parameters $\alpha$ and $\beta$ can be expressed as~\eqref{eqs7}.  Finally,~\eqref{eqs6} is found with the help of~\eqref{Apeq4} by setting $P_{\mathrm{out}}(\gamma_{th}) \approx \widetilde{F}_{g_{\mathrm{FAS}}}\left ( \tfrac{\gamma_{th}}{\overline{\gamma}} \right )$. This completes the proof.

\bibliographystyle{ieeetr}
\bibliography{bibfile}

\end{document}